\documentclass{article}
\usepackage{romp}
\usepackage{amsmath, amssymb, amsthm, mathrsfs}
\usepackage{color}
\newcommand{\e}{\mathrm{e}}

\definecolor{Myred}{cmyk}{0.0,1.0,1.0,0.00}

\title{  A regular analogue of the Smilansky model:
\\[.2em] spectral properties  }
\author{Diana Barseghyan \thanks{ The research has been supported by the Czech
Science Foundation (GA\v{C}R) within the project 17-01706S. D.B.
acknowledges the project SMO ``Pos\'{\i}len\'{\i}
mezin\'arodn\'{\i}ho rozm\v{e}ru v\v{e}deck\'ych aktivit na
P\v{r}\'{\i}rodov\v{e}deck\'{e} fakult\v{e} OU v Ostrav\v{e}'' No.
0924/2016/SaS. Useful comments of the referees are also gratefully
acknowledged.}
 \\ Department of Theoretical Physics, Nuclear Physics Institute CAS, \\
\phantom{c)} 25068 \v{R}e\v{z} near Prague, Czech Republic, \\[.3em]
Department of Mathematics, Faculty of Science, University \\
 \phantom{i)} of Ostrava, 30.~dubna 22, 70103 Ostrava, Czech Republic  \\ e-mail:  \emph{dianabar@ujf.cas.cz, diana.barseghyan@osu.cz} \\[2ex]
          Pavel Exner
\\ Department of Theoretical Physics, Nuclear Physics Institute CAS, \\
\phantom{c)} 25068 \v{R}e\v{z} near Prague, Czech Republic, \\[.3em]Doppler Institute for Mathematical Physics and Applied
Mathematics, \\ Czech Technical University, B\v{r}ehov\'{a} 7, 11519 Prague, Czech Republic                       \\ e-mail: \emph{exner@ujf.cas.cz}  }

\begin{document}

\maketitle
\begin{abstract}
We analyze spectral properties of the operator $H=\frac{\partial^2}{\partial x^2} -\frac{\partial^2}{\partial y^2} +\omega^2y^2-\lambda y^2V(x y)$ in $L^2(\mathbb{R}^2)$, where $\omega\ne 0$ and $V\ge 0$ is a compactly supported and sufficiently regular potential. It is known that the spectrum of $H$ depends on the one-dimensional Schr\"odinger operator $L=-\frac{\mathrm{d}^2}{\mathrm{d}x^2}+\omega^2-\lambda
V(x)$ and it changes substantially as $\inf\sigma(L)$ switches sign. We prove that in the critical case, $\inf\sigma(L)=0$, the spectrum of $H$ is purely essential and covers the interval $[0,\infty)$. In the subcritical case, $\inf\sigma(L)>0$, the essential spectrum starts from $\omega$ and there is a non-void discrete spectrum in the interval $[0,\omega)$. We also derive a bound on the corresponding eigenvalue moments.\end{abstract}


\section{Introduction} \setcounter{equation}{0}

In the paper \cite{Sm04} Uzy Smilansky introduced a simple example
of quantum dynamics which behaves in two substantially different
ways depending on the value of the coupling constant; the original
motivation was to demonstrate that some commonly accepted assumption
in describing irreversible dynamics via coupling to a heat bath can
be avoided. In PDE terms the model is described by the Hamiltonian

\begin{equation}
\label{HSmil}
H_\mathrm{Sm}=-\frac{\partial^2}{\partial x^2}
+\frac12\left( -\frac{\partial^2}{\partial y^2}+y^2 \right) +\lambda
y\delta(x)
\end{equation}
in $L^2(\mathbb{R}^2)$ and the two dynamics types can be expressed
in spectra terms: for $\lambda\le\sqrt{2}$ the operator
\eqref{HSmil} is bounded from below, while for $\lambda>\sqrt{2}$
its spectrum fills the real line \cite{So04}; note that the model
has a mirror symmetry which allows us to consider the situation with $\lambda\ge 0$
only.

The model was subsequently generalized in various way, in
particular, to situations when more than one singular `escape
channel' is open \cite{ES05, NS06}. Other modifications concerned
replacing the oscillator by a potential well of a different shape
\cite{So06b} or by replacing the line by a more general graph
\cite{So06a}. It is also possible to have the motion in the $x$
direction restricted to an interval with periodic boundary
conditions \cite{Gu11, RS07}. In the first named of these papers
time evolution of wave packets was investigated to confirm the idea
that the spectral change in the supercritical case corresponds to
the possibility of an `escape to infinity'; the model was then in
\cite{Gu11} interpreted as a caricature description of a quantum
measurement.

Another question inspired by this work was whether Smilansky model
has an analogue in which the $\delta$-interaction with $y$-dependent
strength is replaced by a smooth potential channel $U$ of increasing
depth. One way to do it is to approximate the $\delta$-interaction
in \eqref{HSmil} by a family of shrinking potentials in the usual
way \cite[Sec.~I.3.2]{AGHH05}. Since the mechanism behind the abrupt
spectral transition is the competition between the
$\delta$-potential eigenvalue $-\frac14 \lambda^2y^2$ and the
oscillator potential, and since the said approximation is based on the assumption that the integral of the potential is preserved at squeezing, we have to match the integral of the approximating potential with the $\delta$ coupling constant, $\int
U(x,y)\,\mathrm{d}x \sim y$, which can be achieved, e.g., by
choosing $U(x,y)=\lambda y^2V(xy)$ for a fixed function $V$. This
suggests the partial differential operator on $L^2(\mathbb{R}^2)$
acting as
\begin{equation}
\label{H1} H=-\frac{\partial^2}{\partial
x^2}-\frac{\partial^2}{\partial y^2} +\omega^2y^2-\lambda y^2V(x y)
\end{equation}
as a suitable candidate for such a regular counterpart to the
operator \eqref{HSmil}. Here $\omega,\, \lambda$ are positive
constants and the potential $V$ with
$\mathrm{supp}\,V\subset[-a,a],\:a>0$, is a nonnegative function
with bounded first derivative; under these assumptions the operator
\eqref{H1} is by Faris-Lavine theorem \cite[Thms.~X.28 and X.38]{RS}
essentially self-adjoint on $C_0^\infty (\mathbb{R}^2)$. Hence it
has a unique self-adjoint extension, namely its closure, which we
will for simplicity also denote by $H$.

In \cite{BE14} we investigated such a model and demonstrated that it
also exhibits an abrupt spectral transition when the coupling
parameter exceeds a critical value. To describe it we need to
introduce a one-dimensional comparison operator,
\begin{equation}\label{comparison}
L=-\frac{\mathrm{d}^2}{\mathrm{d}x^2}+\omega^2-\lambda
V(x)\end{equation}
on $L^2(\mathbb{R})$ with the domain $\mathcal{H}^2(\mathbb{R})$. This allowed
us to characterize different spectral regimes: the operator $H$ is
bounded from below provided that $L$ is non-negative, while if the
spectral threshold of $L$ is negative the spectrum of $H$ fills the
whole real line; for the sake of brevity we shall refer to these
cases as to \emph{(sub)critical} and \emph{supercritical},
respectively.

To make it clear how the operator $L$ arises, consider for a fixed $y$ the corresponding part of (\ref{H1}), namely
\begin{equation}\label{appearance}-\frac{\mathrm{d}^2}{\mathrm{d}
x^2}+\omega^2y^2-\lambda y^2V(x y)\quad\text{on}\;\; \mathbb{R}\,,
\end{equation}
and writing $u(x, y)=\sqrt{|y|} f (z |y|, y)$ we pass from (\ref{appearance}) to the unitarily equivalent operator
$$
y^2 \left(-\frac{\mathrm{d}^2}{\mathrm{d}
z^2}+\omega^2-\lambda V(\pm z)\right)
$$
with (\ref{comparison}) related to the expression inside the bracket.

To be exact, in \cite{BE14} the last term on the right-hand side of
\eqref{H1} was modified by a cut-off factor introduced from
technical reasons which had no influence on the described behavior.
The aim of the present paper is to extend and deepen the analysis of
this operator class in several directions:
\begin{itemize}
\item to analyze the critical case, $\inf\sigma(L)=0$, in
particular, to show that $\sigma(H)=[0,\infty)$ holds in this case
\item in the subcritical case, to show that $\sigma_\mathrm{ess}(H)=
[\omega,\infty)$ and the discrete spectrum is nonempty
\item also in the subcritical case, to derive a bound to eigenvalue
momenta
\end{itemize}
These three topics will be subsequently dealt with in
Sections~\ref{s:crit}--\ref{s:moment} below. Before proceeding to
that, let us mention that there are other systems with narrowing
potential channels which exhibit similar spectral transitions. To
our knowledge, the effect was first noted by M.~Znojil \cite{Zn98}.
Another recent example concerns the operator $-\Delta + |xy|^p -
\lambda (x^2+y^2)^{p/(p+2)}$ in $L^2(\mathbb{R}^2)$ with a fixed
$p\ge 1$ discussed recently in \cite{BEKT16}, where the spectrum
changes from purely discrete to the whole real line at the critical
value of $\lambda$ equal to the principal eigenvalue of the
corresponding one-dimensional anharmonic oscillator.

\section{The critical regime} \label{s:crit}
\subsection{Essential spectrum}\setcounter{equation}{0}

In the free case, $\lambda=0$, the spectrum is purely essential and
equal to $[\omega,\infty)$. We show first that no part of it is lost
when the critical perturbation is switched on, and on the contrary,
the essential spectrum now includes the whole non-negative
half-line.

\begin{theorem} {Theorem}\label{th:critess}
Under the stated assumptions, the essential spectrum of operator $H$
given by (\ref{H1}) contains the half-line $[0, \infty)$ if $\inf
\sigma(L)=0$.
\end{theorem}
\begin{proof}
To prove that any non-negative number $\mu$ belongs to the essential
spectrum of $H$ we are going to employ Weyl's criterion
\cite[Thm.~VII.12]{RS}: we have to find a sequence
$\{\psi_n\}_{n=1}^\infty\subset D(H)$ of unit vectors,
$\|\psi_n\|=1$, which converges weakly to zero and
$$ 
\|H\psi_n-\mu\psi_n\|\to 0 \qquad\text{as}\quad n\to\infty
$$
holds. We are going to use the fact that for any non-negative
potential $V$ which is not identically zero the operator

$$
\widetilde{L}=-\frac{\mathrm{d}^2}{\mathrm{d}x^2}-\lambda V(x)
$$
on $L^2(\mathbb{R})$ has at least one negative eigenvalue
\cite{Si76}, hence the spectral threshold of $L$ is an isolated
eigenvalue; we denote the corresponding normalized eigenfunction by
$h$.

Given a smooth function $\chi$ with $\mathrm{supp}\,\chi\subset[1,
2]$ and satisfying $\int_1^2\chi^2(z)\,\mathrm{d}z=1$, we define
\begin{equation}
\label{sequence}
\psi_n(x,y):=h(x y)\,\e^{i \sqrt{\mu} y}\chi\left(\frac{y}{n}\right)\,,
\end{equation}
where $n\in\mathbb{N}$ is a positive integer to be chosen later. For
the moment we just note that choosing $n$ large enough one can
achieve that $\|\psi_n\|_{L^2(\mathbb{R}^2)}\ge\frac{1}{\sqrt{2}}$ as the
following estimates show,
\begin{eqnarray}
\lefteqn{\int_{\mathbb{R}^2}\left|h(xy)\,\e^{i\sqrt{\mu}
y}\,\chi\left(\frac{y}{n}\right)\right|^2\,
\mathrm{d}x\,\mathrm{d}y= \int_{n}^{2n}\int_{\mathbb{R}}
\left|h(xy)\,\chi\left(\frac{y}{n}\right)
\right|^2\,\mathrm{d}x\,\mathrm{d}y} \nonumber \\ && =
\int_{n}^{2n}\int_{\mathbb{R}}\frac{1}{y}\left|h(t)\,\chi\left(\frac{y}{n}\right)
\right|^2\,\mathrm{d}t\,\mathrm{d}y
=\int_{\mathbb{R}}|h(t)|^2\,\mathrm{d}t\,
\int_{n}^{2n}\frac{1}{y}\left|\chi\left(\frac{y}{n}\right)\right|^2\,\mathrm{d}y
\nonumber \\ && \label{firstpart}
=\int_{n}^{2n}\frac{1}{y}\left|\chi\left(\frac{y}{n}\right)
\right|^2\,\mathrm{d}y=\int_1^2\frac{1}{z}\left|\chi(z)
\right|^2\,\mathrm{d}z\ge\frac{1}{2}\,.
\end{eqnarray}
Our next aim is to show that $\|H\psi_n -\mu\psi_n
\|_{L^2(\mathbb{R}^2)}^2<\varepsilon$ holds for a suitably chosen
$n=n(\varepsilon)$. By a straightforward calculation one finds
$$
\frac{\partial^2\psi_n}{\partial x^2}
=y^2h''(x y)\,\e^{i \sqrt{\mu} y}\chi\left(\frac{y}{n}
\right)
$$
and
\begin{eqnarray}\nonumber
\lefteqn{\frac{\partial^2\psi_n}{\partial y^2} =x^2h''(x y)\,\e^{i
\sqrt{\mu} y}\chi\left(\frac{y}{n}\right) +2ix\sqrt{\mu}\, h'(x
y)\,\e^{i \sqrt{\mu} y}\chi\left(\frac{y}{n}\right)
+\frac{2x}{n}h'(x y)\,\e^{i \sqrt{\mu}
y}\chi'\left(\frac{y}{n}\right)} \\ &&-\mu h(x y)\,\e^{i \sqrt{\mu}
y}\chi\left(\frac{y}{n}\right) +2\frac{i \sqrt{\mu}}{n}h(x y)\,\e^{i
\sqrt{\mu} y}\chi'\left(\frac{y}{n}\right) +\frac{1}{n^2}h(x
y)\,\e^{i \sqrt{\mu} y}\chi''\left(\frac{y}{n}\right)\,.
\label{calculations}
\end{eqnarray}
We need to show that choosing $n$ sufficiently large one can make
the terms on the right-hand side of (\ref{calculations}) as small as
we wish. Changing the integration variables, we get the following
estimate,
\begin{eqnarray*}
\lefteqn{\int_{\mathbb{R}^2}\left|x^2\,h''(x y)\,\e^{i \sqrt{\mu}
y}\chi\left(\frac{y}{n}
\right)\right|^2\,\mathrm{d}x\,\mathrm{d}y=\int_{n}^{2n}\int_{\mathbb{R}}
\left|x^2\,h''(x y)\, \chi\left(\frac{y}{n}\right)
\right|^2\,\mathrm{d}x\,\mathrm{d}y} \\ &&
=\int_{n}^{2n}\frac{1}{y^5}
\left|\chi\left(\frac{y}{n}\right)\right|^2\,\mathrm{d}y
\,\int_{\mathbb{R}}t^4|h''(t)|^2\,\mathrm{d}t \le\frac{1}{n^4}
\int_1^2|\chi(z)|^2\mathrm{d}z\,\int_{\mathbb{R}}t^4|h''(t)|^2\,\mathrm{d}t\,;
\end{eqnarray*}
note that since the potential $V$ has by assumption a compact
support, the ground state eigenfunction $h$ decays exponentially as
$|x|\to\infty$, hence the second integral in the last expression
converges. In the same way we establish the remaining inequalities
we need,
\begin{eqnarray*}
\int_{\mathbb{R}^2}\left|x h'(x y)\,\e^{i \sqrt{\mu}
y}\chi\left(\frac{y}{n}\right)
\right|^2\,\mathrm{d}x\,\mathrm{d}y\\=\int_{n}^{2n}
\int_{\mathbb{R}}\frac{1}{y^3}\left|t
h(t)\,\chi\left(\frac{y}{n}\right)
\right|^2\,\mathrm{d}t\,\mathrm{d}y
\le\frac{1}{n^2}\int_1^2|\chi(z)|^2
\mathrm{d}z\,\int_{\mathbb{R}}t^2|h(t)|^2\,\mathrm{d}t\,,
\end{eqnarray*}
\begin{eqnarray*}
\int_{\mathbb{R}^2}\left|\frac{x}{n}\,h'(x y)\,\e^{i
\sqrt{\mu} y}\chi'\left(\frac{y}{n}\right)
\right|^2\,\mathrm{d}x\,\mathrm{d}y\\=\int_{n}^{2n}\int_{\mathbb{R}}
\frac{1}{y}\left|\frac{t}{n
y}h'(t)\,\chi'\left(\frac{y}{n}\right)\right|^2
\,\mathrm{d}t\,\mathrm{d}y\le\frac{1}{n^4}
\int_1^2|\chi'(z)|^2\mathrm{d}z\,\int_{\mathbb{R}}t^2|h'(t)|^2\,\mathrm{d}t\,,
\end{eqnarray*}
\begin{eqnarray*}
\int_{\mathbb{R}^2}\left|\frac{1}{n}\,h(x y)\,\e^{i
\sqrt{\mu} y}\chi'\left(\frac{y}{n}\right)
\right|^2\,\mathrm{d}x\,\mathrm{d}y\\=\frac{1}{n^2}\int_{n}^{2n}\int_{\mathbb{R}}
\frac{1}{y}\left|h(t)\,\chi'\left(\frac{y}{n}\right)\right|^2
\,\mathrm{d}t\,\mathrm{d}y \le\frac{1}{n^2}
\int_1^2\frac{|\chi'(z)|^2}{z}\mathrm{d}z\,\int_{\mathbb{R}}|h(t)|^2\,\mathrm{d}t\,,
\end{eqnarray*}
\begin{eqnarray*}
\int_{\mathbb{R}^2}\left|\frac{1}{n^2}h(x y)\,\e^{i
\sqrt{\mu} y}\chi''\left(\frac{y}{n}
\right)\right|^2\,\mathrm{d}x\,\mathrm{d}y\\=\frac{1}{n^4}\int_{n}^{2n}
\int_{\mathbb{R}}\frac{1}{y}\left|h(t)\,\chi''\left(\frac{y}{n}\right)
\right|^2\,\mathrm{d}t\,\mathrm{d}y
\le\frac{1}{n^4}\int_1^2|\chi''(z)|^2\mathrm{d}z
\,\int_{\mathbb{R}}|h(t)|^2\,\mathrm{d}t\,,
\end{eqnarray*}
which show that the corresponding terms are either
$\mathcal{O}(n^{-2})$ or $\mathcal{O}(n^{-4})$ as $n\to\infty$, hence
choosing $n$ large enough we can achieve that the sum of all the
integrals at the left-hand sides of the above inequalities is less
than $\varepsilon$. This allows us to estimate the expression in
question,
\begin{eqnarray*}
\lefteqn{\int_{\mathbb{R}^2}|H\psi_n-\mu
\psi_n|^2(x,y)\,\mathrm{d}x\,\mathrm{d}y} \\ &&
=\int_{\mathbb{R}^2}\left|-\frac{\partial^2\psi_n}{\partial x^2}-
\frac{\partial^2\psi_n}{\partial y^2}+\omega^2y^2\psi_n-\lambda
y^2V(x y)\psi_n-\mu \psi_n\right|^2\, \mathrm{d}x\,\mathrm{d}y \\
&& \le\int_{n}^{2n} \int_{\mathbb{R}}\biggl|y^2h''(x
y)\chi\left(\frac{y} {n}\right)-\omega^2\,y^2h(x
y)\chi\left(\frac{y}{n}\right) +\lambda y^2\,V(x y)h(x
y)\chi\left(\frac{y}{n}\right)\biggr|^2\,\mathrm{d}x\,
\mathrm{d}y+\varepsilon \\ &&
=\int_{n}^{2n}\int_{\mathbb{R}}\biggl|y^2\left(h''(xy)-\omega^2h(xy)+\lambda
V(xy)h(xy)\right)\chi\left(\frac{y}{n}\right)\biggr|^2
\,\mathrm{d}x\,\mathrm{d}y+\varepsilon\,,
\end{eqnarray*}
and using the fact that $Lh=0$ holds by assumption, the last
inequality implies
\begin{equation}\label{fin.}
\int_{\mathbb{R}^2}|H\psi_n-\mu\psi_n|^2(x,y)\,
\mathrm{d}x\,\mathrm{d}y<\varepsilon\,.
\end{equation}
To complete the proof we fix a sequence
$\{\varepsilon_j\}_{j=1}^\infty$ such that $\varepsilon_j\searrow0$
holds as $j\to\infty$, and to any $j$ we construct a function
$\psi_{n(\varepsilon_j)}$ according to \eqref{sequence} with the
parameters chosen in such a way that $n(\varepsilon_j) >2
n(\varepsilon_{j-1})$. The norms of $H\psi_{n(\varepsilon_j)}$
satisfy inequality (\ref{fin.}) with $\varepsilon_j$ on the
right-hand side, and since the supports of
$\psi_{n(\varepsilon_j)},\:j =1,2,\ldots,$ do not intersect each
other by construction, the sequence of these functions converges
weakly to zero. The same is true for the sequence of unit vectors
$\tilde\psi_{n(\varepsilon_j)} :=
\frac{\psi_{n(\varepsilon_j)}}{\|\psi_{n(\varepsilon_j)}\|}$ and the
norms of $H\tilde\psi_{n(\varepsilon_j)}$ satisfy an inequality
similar to (\ref{fin.}), this time with $2\varepsilon_j$ on the
right-hand side; this yields the sought claim.
\end{proof}
\begin{remark} {Remark}
Note that he sequence $\psi_{n(\varepsilon_j)}$ can be ``less
sparse'' then we have assumed in the proof. Indeed, these functions
tend to zero uniformly on compact sets due to the factor
$\frac{1}{\sqrt{k}}$ in their definition, hence for any compactly
supported test function $\phi$ one has
$\int_{\mathbb{R}^2}\psi_{n(\varepsilon_j)}
\overline{\phi}\,\mathrm{d}x\,\mathrm{d}y\to0$ as $k\to\infty$. As
compactly supported functions are dense in $L^2(\mathbb{R}^2)$, this
is sufficient to ensure the weak convergence.
 \end{remark}

\subsection{Non-negativeness}

Now we are going to show that under our assumptions the operator $H$
has no negative spectrum in the critical regime.

\begin{theorem} {Theorem}\label{th:nonneg}
Let $\inf\sigma(L)=0$, then $H$ is non-negative.
\end{theorem}
\begin{proof}
For any $u\in \mathrm{dom}(Q_H)$, the quadratic form associated with
$H$, we have
$$
Q_H[u]=\int_{\mathbb{R}^2}\left|\frac{\partial{u}}{\partial{x}}\right|^2\,
\mathrm{d}x\,\mathrm{d}y+\int_{\mathbb{R}^2}\left|
\frac{\partial{u}}{\partial{y}}\right|^2\,\mathrm{d}x\,\mathrm{d}y+\omega^2\int_{\mathbb{R}^2}y^2
|u|^2\,\mathrm{d}x\,\mathrm{d}y-\lambda \int_{\mathbb{R}^2}y^2 V(x
y)|u|^2\,\mathrm{d}x\,\mathrm{d}y\,.
$$
Neglecting the second term on the right-hand side, we can estimate
the form value as
\begin{equation}\label{q}
Q_H[u] \ge
\int_{\mathbb{R}}\left(\int_{\mathbb{R}}\left|\frac{\partial{u}}{\partial{x}}\right|^2\,
\mathrm{d}x+\omega^2\int_{\mathbb{R}}y^2
|u|^2\,\mathrm{d}x-\lambda \int_{\mathbb{R}}y^2 V(x y)
|u|^2\,\mathrm{d}x\right)\,\mathrm{d}y\,.
\end{equation}
For any fixed $y\neq0$ we change variables in the inner integral on
the right-hand side and denote $w(t,y)=u\left(\frac{t}{y},
y\right)$. Using the fact that $L\ge 0$ one finds
\begin{eqnarray*}
\lefteqn{\int_{\mathbb{R}}\left|\frac{\partial{u}}{\partial{x}}\right|^2(x,
y)\,\mathrm{d}y+\omega^2\, y^2\int_{\mathbb{R}} |u|^2(x,
y)\,\mathrm{d}x-\lambda y^2\int_{\mathbb{R}} V(x y) |u|(x,
y)^2\,\mathrm{d}x}
\\  && =\frac{1}{|y|}\left(y^2\int_{\mathbb{R}}\left|\frac{\partial{w}}{\partial{t}}\right|^2(t,
y)\,\mathrm{d}t+\omega^2\,y^2 \int_{\mathbb{R}} |w|^2(t,
y)\,\mathrm{d}t-\lambda y^2\int_{\mathbb{R}}V(t) |w|^2(t,
y)\,\mathrm{d}t\right)
\\  && =|y|\left(\int_{\mathbb{R}}\left|\frac{\partial{w}}{\partial{t}}\right|^2(t,
y)\,\mathrm{d}t+\omega^2\,\int_{\mathbb{R}} |w|^2(t,
y)\,\mathrm{d}t-\lambda \int_{\mathbb{R}}V(t) |w|^2(t,
y)\,\mathrm{d}t\right)\ge 0\,,
\end{eqnarray*}
which in combination with the inequality (\ref{q}) establishes our
claim.
\end{proof}

\begin{corollary}{Corollary}
In the critical case we have $\sigma(H)= \sigma_\mathrm{ess}(H)
=[0,\infty)$.
\end{corollary}

\section{Subcritical regime} \label{s:subrit}
\subsection{Essential spectrum}\setcounter{equation}{0}

In contrast to the critical case, one can now guarantee only that
the perturbation does not make the essential spectrum  shrink.
\begin{theorem}{Theorem} \label{th:subcritess}
Let $\inf\,\sigma(L)>0$ then $\sigma_{\mathrm{ess}}
(H)\supset[\omega, \infty)$.
\end{theorem}
\begin{proof}
As before we are going to construct a Weyl sequence for any number
$\mu\ge\omega$. This time we employ the functions
\begin{equation}
\varphi_k(x, y)=\frac{1}{\sqrt{k}}\,g(y)\, e^{i \sqrt{\mu-\omega}
x}\, \eta\left(\frac{x}{k}\right)\,,
\end{equation}
where $g$ is the normalized eigenfunction associated with the
principal eigenvalue of the harmonic oscillator,
$h_{\mathrm{osc}}=-\frac{\mathrm{d}^2}{\mathrm{d}y^2}+\omega^2 y^2$
on $L^2(\mathbb{R})$, the function $\eta\in C_0^\infty(1, 2)$ is
supposed to satisfy the following condition,
$$
\int_1^2\eta^2(z)\,\mathrm{d}z=1\,,
$$
and $k\in\mathbb{N}$ is a positive integer to be chosen later. Let
us note that $\|\varphi_k\|_{L^2(\mathbb{R}^2)}=1$ because
\begin{equation}
\int_{\mathbb{R}^2} |\varphi_k(x,y)|^2\, \mathrm{d}x\,\mathrm{d}y
=\frac{1}{k}\int_{\mathbb{R}}\left|g(y)\right|^2\,\mathrm{d}y\,
\int_k^{2k}\left|\eta\left(\frac{x}{k}\right)\right|^2\,\mathrm{d}x
=\int_{\mathbb{R}}g^2(y)\,\mathrm{d}y\,\int_1^2\eta^2(z)\,\mathrm{d}z=1\,.
\end{equation}
Our aim is to show that $\|H\varphi_k
-\mu\varphi_k\|_{L^2(\mathbb{R}^2)}^2<\varepsilon$  holds for an
appropriate $k=k(\varepsilon)$. By a straightforward calculation one
gets
 \begin{equation}\label{calculations*}
\frac{\partial^2\varphi_k}{\partial x^2} =
\left(-\frac{(\mu-\omega)}{\sqrt{k}} \,\eta\left(\frac{x}{k}\right)
+\frac{2i \sqrt{\mu-\omega}}{k\sqrt{k}}
\,\eta^\prime\left(\frac{x}{k}\right) +\frac{1}{k^2\sqrt{k}}
\,\eta''\left(\frac{x}{k}\right)\right) g(y)\,\e^{i
\sqrt{\mu-\omega}x}
 \end{equation}
and
$$
\frac{\partial^2\varphi_k}{\partial y^2} =\frac{1}{\sqrt{k}}\,g''(y)
\,\e^{i \sqrt{\mu-\omega} x}\eta\left(\frac{x}{k}\right)\,.
$$
We want to show that choosing $k$ sufficiently large one can make a
part of the terms at the right-hand side of (\ref{calculations*}) as
small as one wishes. Changing the integration variables, we get the
following estimates
\begin{eqnarray*}
\lefteqn{\int_{\mathbb{R}^2}\left|\frac{1}{k\sqrt{k}} \,g(y)\, \e^{i
\sqrt{\mu-\omega}
x}\eta^\prime\left(\frac{x}{k}\right)\right|^2\,\mathrm{d}x\,\mathrm{d}y
=\frac{1}{k^3}\int_{k}^{2k}\int_{\mathbb{R}} \left|g(y)
\,\eta^\prime\left(\frac{x}{k}\right)\right|^2\,\mathrm{d}x\,\mathrm{d}y}
\\ &&
\le\frac{1}{k^3}\int_{k}^{2k}\left|\eta^\prime\left(\frac{x}{k}\right)\right|^2\,\mathrm{d}x
\,\int_{\mathbb{R}}|g(y)|^2\,\mathrm{d}y \le\frac{1}{k^2}
\int_1^2|\eta^\prime(z)|^2\mathrm{d}z\,\int_{\mathbb{R}}|g(y)|^2\,\mathrm{d}y\,,
\end{eqnarray*}
and in the same way we establish the remaining inequality needed to
demonstrate our claim,
\begin{eqnarray*}
\lefteqn{\int_{\mathbb{R}^2}\left|\frac{1}{k^2\sqrt{k}}
\,g(y)\,\e^{i \sqrt{\mu-\omega}
x}\eta''\left(\frac{x}{k}\right)\right|^2\,\mathrm{d}x\,\mathrm{d}y=\frac{1}{k^5}\int_{k}^{2k}\int_{\mathbb{R}}
\left|g(y) \,\eta''\left(\frac{x}{k}\right)\right|^2
\,\mathrm{d}x\,\mathrm{d}y} \\ &&
\le\frac{1}{k^4}\,\,\int_{\mathbb{R}}|g(y)|^2\,\mathrm{d}y\,
\int_1^2|\eta''(z)|^2\mathrm{d}z\,. \phantom{AAAAAAAAAAAAAAAAAAAAA}
\end{eqnarray*}
Consequently, choosing $k$ large enough one can achieve that the
integrals on the left-hand sides of the above inequalities will be
less than $\varepsilon$, which implies
\begin{eqnarray*}
\lefteqn{\int_{\mathbb{R}^2}|H\varphi_k-\mu
\varphi_k|^2(x,y)\,\mathrm{d}x\,\mathrm{d}y} \\ &&
=\int_{\mathbb{R}^2}\left|-\frac{\partial^2\varphi_k}{\partial x^2}-
\frac{\partial^2\varphi_k}{\partial
y^2}+\omega^2y^2\varphi_k-\lambda y^2V(x y)\varphi_k-\mu
\varphi_k\right|^2\, \mathrm{d}x\,\mathrm{d}y
\\  && \le\frac{1}{k}\int_k^{2k}
\int_{\mathbb{R}}\biggl|-g''(y) +(\mu-\omega) g(y) +\omega^2 y^2
g(y)-\lambda y^2 V(x y) g(y)-\mu  g(y) \biggr|^2\, \\ &&
\times \eta\left(\frac{x}{k}\right)\,\mathrm{d}x\,\mathrm{d}y+\varepsilon\,.
\end{eqnarray*}
Using now the fact that $g$ is the ground-state eigenfunction of
$h_{\mathrm{osc}}$ and that the potential $V$ is compactly
supported, the above result implies
\begin{eqnarray*}
\lefteqn{\int_{\mathbb{R}^2}|H\varphi_k-\mu
\varphi_k|^2(x,y)\,\mathrm{d}x\,\mathrm{d}y
\le\frac{\lambda^2}{k}\int_k^{2k}\int_{\mathbb{R}} y^4 V^2(x y)
\,g^2(y)
\,\eta^2\left(\frac{x}{k}\right)\,\mathrm{d}x\,\mathrm{d}y+\varepsilon}
\\ &&
\le\frac{\lambda^2}{k}\int_k^{2k}\int_{-\frac{a}{k}}^{\frac{a}{k}}
y^4 V^2(x y) \,g^2(y)
\,\eta^2\left(\frac{x}{k}\right)\,\mathrm{d}x\,\mathrm{d}y+\varepsilon
\\ &&
\le\frac{a^4 \lambda^2
\|V\|_\infty^2}{k^5}\int_{-\frac{a}{k}}^{\frac{a}{k}}g^2(y)
\,\mathrm{d}y\, \int_k^{2k}
\eta^2\left(\frac{x}{k}\right)\,\mathrm{d}x+\varepsilon
\\ &&
\le\frac{a^4 \lambda^2
\|V\|_\infty^2}{k^4}\,\int_{\mathbb{R}}g^2(y)\,\mathrm{d}y\,
\int_1^2 \eta^2(z)\,\mathrm{d}z+\varepsilon\,, \phantom{AAAAAAAAAAAAAAAA}
\end{eqnarray*}
and consequently, for a large enough $k$ we have
\begin{equation}\label{final}
\int_{\mathbb{R}^2}|H\varphi_k-\mu\varphi_k|^2(x,y)\:
\mathrm{d}x\,\mathrm{d}y<2\varepsilon\,.
\end{equation}
To complete the proof we proceed as in Theorem~\ref{th:critess}
choosing a sequence $\{\varepsilon_j\}_{j=1}^\infty$ such that
$\varepsilon_j\searrow0$ holds as $j\to\infty$ and to any $j$ we
construct a function $\varphi_{k(\varepsilon_j)}$ with the
parameters chosen in such a way that $k(\varepsilon_j) >2
k(\varepsilon_{j-1})$. The norms of $H\varphi_{k(\varepsilon_j)}$
satisfy the inequality (\ref{final}) with $2\varepsilon_j$ on the
right-hand side, and the sequence converges by construction weakly
to zero; this time the elements of the sequence are already
normalized.
\end{proof}

\subsection{Discrete spectrum} \label{s:discr}

Next we are going to show that a subcritical perturbation cannot
inflate the essential spectrum.

\begin{theorem} {Theorem}\label{th:disc}
Let $\inf \sigma(L)>0$, then the spectrum of operator $H$ below
$\omega$ is discrete.
\end{theorem}
\begin{proof}
We employ a Neumann bracketing and the minimax principle
\cite[Secs.~XIII.1 and XIII.15]{RS}. Let us fix a natural number
$k$, later to be chosen large, and let $h^{(\pm)}_{n, k}$ and $h_k$
be the Neumann restrictions of operator $H$ to the regions
$$
G^{(\pm)}_{n, k}=\{x:\:|x|\le k\} \times \left\{y:\: 1+\ln n<\pm
y\le 1+\ln(n+1)\right\}
$$
and
$$
G_k=\{|x|>k\}\times \mathbb{R}\,, \;\; G^{(0)}=[-k,k]\times[-1,1]\,.
$$
We have the inequality
\begin{equation}\label{N}
H\ge\bigoplus_{n=1}^\infty\:
\left(h^{(+)}_{n, k}\oplus h_{n, k}^{(-)} \right)\oplus h_k \oplus h^{(0)}\,.
\end{equation}
Since the spectrum of $h^{(0)}$ is obviously discrete, to prove our claim we first demonstrate that the spectral thresholds of $h_{n, k}^{(\pm)}$ tend for large enough $k$ to infinity as $n\to\infty$, and secondly, that for any $\Lambda<\omega$ one can choose $k$ in such a way that the spectrum of $h_k$ below $\Lambda$ is empty. Since the function $V$ has a bounded derivative and is compactly supported we have
$$
V(x y)-V(x (1+\ln n))=\mathcal{O}\left(\frac{1}{n\ln n}\right)\,,
\quad y^2-(1+\ln n)^2=\mathcal{O}\left(\frac{\ln n}{n}\right)
$$
for any $(x,y) \in G_{n, k}^{(+)}$ and similar relation for
$G^{(-)}_{n, k}$. This yields
$$
y^2V(xy)-(1+\ln n)^2\,V(\pm x(1+\ln n))=\mathcal{O}\left(\frac{\ln
n}{n}\right)
$$
for any $(x,y) \in G_{n, k}^{(\pm)}$, which further implies the
asymptotic inequalities
\begin{equation} \label{l_nk}
\inf\sigma(h_{n, k}^{(\pm)})\ge\inf\sigma(l_{n,
k}^{(\pm)})+\mathcal{O}\left(\frac{\ln n}{n}\right)\,,
\end{equation}
in which the Neumann operators $l_{n,k}^{(\pm)}:=-\frac{\partial^2}{\partial x^2}-\frac{\partial^2}{\partial y^2}+\omega^2(1+\ln n)^2-\lambda (1+\ln
n)^2\,V(\pm x(1+\ln n))$ on $G^{(\pm)}_{n, k}$ have separated variables. Since the principal eigenvalue of $-\frac{\mathrm{d}^2}{\mathrm{d}y^2}$ on an interval with Neumann boundary conditions is zero, we have
\begin{equation}\label{widetildel}
\inf\sigma\left(l^{(\pm)}_{n,k}\right)
=\inf\sigma\big(\widetilde{l}_{n, k}\big)\,,
\end{equation}
where
$$
\widetilde{l}_{n, k} :=-\frac{\mathrm{d}^2}{\mathrm{d}x^2}
+\omega^2\,(1+\ln n)^2-\lambda (1+\ln n)^2\,V(x(1+\ln n))
$$
acts on $L^2(-k, k)$. To proceed with the proof we employ the following result established in \cite{BEWZ93}:
Let $l_k=-\frac{\mathrm{d}^2}{\mathrm{d}x^2}+\omega^2-\lambda V(x)$
be the Neumann restriction of operator $L$ given by
\eqref{comparison} to the interval $[-k, k],\;k>0$, then we have
\begin{equation}\label{hnp1}
\inf\sigma\left(l_k\right)\to\gamma_0\;\quad\text{as}\quad
k\to\infty\,,
\end{equation}
where $\gamma_0:= \inf\sigma(L)$. Further we note that by the change of the variable $x=\frac{t}{1+\ln n}$ the operator $\widetilde{l}_{n, k}$ is unitarily equivalent to $(1+\ln n)^2\,L_{n, k},$ where $L_{n, k}=-\frac{\partial^2}{\partial x^2}+\omega^2-\lambda V$ in $L^2(-k (1+\ln n), \,k (1+\ln n))$ with Neumann conditions at the endpoints of the interval. Then in view of inequalities (\ref{l_nk})-(\ref{widetildel}), the relation $\inf\sigma(\widetilde{l}_{n, k}) = (1+\ln n)^2\, \inf\sigma(L_{n, k})$, and (\ref{hnp1}) we conclude the proof of the discreteness of $\bigoplus_{n=1}^\infty\:\left(h^{(+)}_{n, k}\oplus h_{n, k}^{(-)} \right)$.

It remains to inspect the spectrum of $h_k$. Since $V$ is compactly supported then $V(xy)=0$ if $|x|>k$ and $|y|>\frac{a}{k}$, hence in view of \eqref{H1} we have $h_k=-\Delta+\omega^2 y^2+\mathcal{O}(k^{-2})$, and therefore
\begin{equation}\label{perturb.}
\inf \sigma(h_k)=\inf \sigma(-\Delta+\omega^2 y^2)+\mathcal{O}(k^{-2})
\end{equation}
by an elementary perturbation argument \cite{K95}. Since the operator $-\Delta+\omega^2 y^2$ allows for separation of variables, which shows that its spectrum is $[\omega,\infty)$, in combination (\ref{perturb.}) we arrive at
\begin{equation}\label{omega}
\inf \sigma(h_k)=\omega+\mathcal{O}(k^{-2})\,,
\end{equation}
which concludes the proof of Theorem~\ref{th:disc}.
\end{proof}

\begin{remark}{Remark}
{\rm Let us denote now the operator \eqref{H1} as $H_\lambda$. Since the potential $V$ is non-negative it is easy to see that the relations $\mathrm{dom}(Q_{H_\lambda}) \subset \mathrm{dom}(Q_{H_\mu})$ and  $Q_{H_\mu}\le Q_{H_\lambda}$ hold provided $\lambda\le\mu$, in other words, that we have operator inequality $H_\mu\le H_\lambda$. This allows us to localize better the discrete spectrum.}
\end{remark}

\begin{corollary}{Corollary}
The discrete spectrum of a subcritical operator \eqref{H1} is contained in the interval $[0,\omega)$.
\end{corollary}
\begin{proof}
By the previous remark we have $H_\lambda\ge H_{\lambda_\mathrm{crit}}$ which yields the claim in combination with Theorem~\ref{th:nonneg} and the minimax principle.
\end{proof}

\subsection{Existence of the discrete spectrum}

The above results, on the other hand, tell us nothing about the existence of the discrete spectrum. This is the question we are going to address now.

\begin{theorem}{Theorem}
Let $\inf \sigma (L)>0$, then the discrete spectrum of $H$ is non-empty.
\end{theorem}
\begin{proof}
In view of Theorem~\ref{th:disc} it is sufficient to construct a normalized trial function $\phi$ such that the corresponding value of the quadratic form $Q_H$ is less than $\omega$. This time we use the letter $h$ to denote the normalized ground-state eigenfunction of the one-dimensional harmonic oscillator governed by $h_{\mathrm{osc}}=-\frac{\mathrm{d}^2}{\mathrm{d}y^2}+\omega^2y^2$ on $L^2(\mathbb{R})$ and set
$$
\phi(x, y):=\frac{1}{\sqrt{k}} h(y)
\chi\left(\frac{x}{k}\right)\,,
$$
where $\chi(z)$ is a real-valued smooth function with $\mathrm{supp} (\chi)=[-1, 1]$ such that
$$
\int_{-1}^1\chi^2(z)\,\mathrm{d}z=1\,,\quad\min_{|z|\le1/2}\,
\chi(z)=:\alpha>0\,,
$$
and $k$ is a natural number to be chosen later. A straightforward computation yields
\begin{eqnarray}\nonumber
\lefteqn{Q_H[\phi]=\int_{\mathbb{R}^2}\left|\frac{\partial{\phi}}{\partial{x}}\right|^2\,
\mathrm{d}x\,\mathrm{d}y
+\int_{\mathbb{R}^2}\left|\frac{\partial{\phi}}{\partial{y}}\right|^2\,
\mathrm{d}x\,\mathrm{d}y +\int_{\mathbb{R}^2}\omega^2 y^2\,
|\phi|^2\,\mathrm{d}x\,\mathrm{d}y} \\ && \nonumber \qquad - \lambda \int_{\mathbb{R}^2}y^2
V(xy)\,|\phi^2|\,\mathrm{d}x\,\mathrm{d}y \\ && \nonumber
=\frac{1}{k^3}\int_{\mathbb{R}^2}h^2(y)\,(\chi^\prime)^2
\left(\frac{x}{k}\right)\,\mathrm{d}x\,\mathrm{d}y
+\frac{1}{k}\int_{\mathbb{R}^2}\left(h^\prime\right)^2(y)\,
\chi^2\left(\frac{x}{k}\right)\,\mathrm{d}x\,\mathrm{d}y
\\  && \nonumber \qquad +\frac{1}{k}\int_{\mathbb{R}^2}\omega^2 y^2
\,h^2(y)\,\chi^2\left(\frac{x}{k}\right)\,\mathrm{d}x\,\mathrm{d}y
-\frac{\lambda}{k}\int_{\mathbb{R}^2}y^2 V(xy)\, h^2(y)\,\chi^2
\left(\frac{x}{k}\right)\,\mathrm{d}x\,\mathrm{d}y
\\ && \nonumber
=\mathcal{O}\left(\frac{1}{k^2}\right)+\frac{1}{k}
\int_{\mathbb{R}^2}\left(\left(h^\prime\right)^2(y)
+\omega^2y^2\,h^2(y)\right)\,\chi^2\left(\frac{x}{k}\right)\,
\mathrm{d}x\,\mathrm{d}y
\\ && \nonumber \qquad -\frac{\lambda}{k}\int_{\mathbb{R}^2}y^2
V(xy)\, h^2(y)\, \chi^2\left(\frac{x}{k}\right)\,
\mathrm{d}x\,\mathrm{d}y
\\ && \nonumber
=\mathcal{O}\left(\frac{1}{k^2}\right)+\frac{\omega}{k}\int_{\mathbb{R}^2}
h^2(y)\,\chi^2\left(\frac{x}{k}\right)\,\mathrm{d}x\,\mathrm{d}y
-\frac{\lambda}{k}\int_{\mathbb{R}^2}y^2 V(xy)\, h^2(y)\,
\chi^2\left(\frac{x}{k}\ \right)\,\mathrm{d}x\,\mathrm{d}y
\\ && \label{non-emptiness}
=\mathcal{O}\left(\frac{1}{k^2}\right)+\omega-\frac{\lambda}{k}\int_{\mathbb{R}^2}y^2
V(xy)\, h^2(y)\,
\chi^2\left(\frac{x}{k}\right)\,\mathrm{d}x\,\mathrm{d}y\,.
\end{eqnarray}
We need to estimate the last term on the right-hand side of (\ref{non-emptiness}). One has
\begin{eqnarray*}
\lefteqn{\frac{\lambda}{k}\int_{\mathbb{R}^2}y^2 V(x y)\, h^2(y)\, \chi^2\left(\frac{x}{k}\right)\,
\mathrm{d}x\,\mathrm{d}y=\frac{\lambda}{k}\int_{-k}^k
\int_{\mathbb{R}} y^2 V(x y)\, h^2(y)\, \chi^2\left(\frac{x}{k}\right)\,\mathrm{d}x\,\mathrm{d}y} \\ &&
\ge\frac{\lambda}{k}\int_{-k/2}^{k/2} \int_0^\infty y^2 V(x y)\, h^2(y)\,
\chi^2\left(\frac{x}{k}\right)\,\mathrm{d}x\,\mathrm{d}y\ge\frac{\alpha^2
\lambda}{k}\int_{-k/2}^{k/2} \int_0^\infty y^2 V(x y)\,
h^2(y)\,\mathrm{d}x\,\mathrm{d}y\\ && =\frac{\alpha^2
\lambda}{k}\int_0^\infty\int_{-k y/2}^{k y/2}y\, V(t)\,
h^2(y)\,\mathrm{d}t\,\mathrm{d}y\ge\frac{\alpha^2
\lambda}{k}\int_1^\infty\int_{-k/2}^{k/2}y V(t)\,
h^2(y)\,\mathrm{d}t\,\mathrm{d}y\\ && \ge\frac{\alpha^2
\lambda}{k}\int_1^\infty y
h^2(y)\,\mathrm{d}y\,\int_{-k/2}^{k/2}V(t)\,\mathrm{d}t\,.
\end{eqnarray*}
If $k$ is large enough then the above estimate implies
$$
\frac{\lambda}{k}\int_{\mathbb{R}^2}y^2 V(x y)\, h^2(y)\,
\chi^2\left(\frac{x}{k}\right)\,\mathrm{d}x\,\mathrm{d}y\ge\frac{\alpha^2
\lambda}{k}\int_1^\infty y
h^2(y)\,\mathrm{d}y\,\int_{-a}^aV(t)\,\mathrm{d}t\,,
$$
hence in combination with (\ref{non-emptiness}) we infer that
$$
Q_H[\phi] \le \mathcal{O}\left(\frac{1}{k^2}\right)+\omega-\frac{\alpha^2
\lambda}{k}\int_1^\infty y
h^2(y)\,\mathrm{d}y\,\int_{-a}^aV(t)\,\mathrm{d}t<\omega\,,
$$
which is what we set out to demonstrate.
\end{proof}

\section{Eigenvalue estimates} \label{s:moment}
\setcounter{equation}{0}
Since the spectrum of $H$ in $[0,\omega)$ is non-empty and consists of the discrete eigenvalues of finite multiplicity one can think about the eigenvalue momentum estimates in the spirit of Lieb and Thirring \cite{LT76}. To state our result we need the following definition:

\smallskip

\noindent Let $\inf \sigma (L)=\gamma_0>0$ and  let $l_k$ be the Neumann restriction of $L$ to the interval  $[-k, k],\;k>0$. We denote
 \begin{equation}\label{kappa}
 \kappa:=\min \left\{k: \,\inf \sigma (l_k)\ge\gamma_0/2\right\}
 \end{equation}
and observe that (\ref{hnp1}) guarantees the existence of such a number.

\smallskip

\noindent Then we can make the following claim:
\begin{theorem}{Theorem}
Let $\inf \sigma (L)=\gamma_0>0$, then for any $\sigma>\frac{1}{2}$ the inequality
\begin{eqnarray*}
\mathrm{tr} (\omega-H)_+^\sigma \le2\lambda^{2\sigma} \|V\|_\infty^{2\sigma}
a^{4\sigma}\sum_{n=1}^\infty\frac{1}{\alpha_1^{2\sigma}\left(\sqrt{\lambda
\|V\|_\infty} a+(n-1) \pi\right)^{2\sigma}}\,\\
+\, \left(\frac{2\alpha_1 \sqrt{\omega+\lambda \alpha_1^2 \|V\|_\infty}}{\pi}+1\right)^2 \left(\omega+\lambda \alpha_1^2 \|V\|_\infty\right)^\sigma
\end{eqnarray*}
holds, where
\begin{equation}\label{alpha.}
\alpha_1:=\max\left\{\sqrt{\kappa},\,\frac{2\omega}{\gamma_0},\,\frac{\sqrt{\lambda \|V\|_\infty} a}{\sqrt{2 \omega}}\right\}
\end{equation}
with $\kappa$ defined by (\ref{kappa}).
\end{theorem}
\begin{proof}
We are going to employ a bracketing argument similar to that used in Subsection~\ref{s:discr} imposing additional Neumann conditions at the boundaries of the regions
$$
G^{(\pm)}_n=\left\{-\alpha_1<x<\alpha_1\right\}\times \left\{\alpha_n< \pm y<\alpha_{n+1}\right\}\,,
$$
$$
Q_n^{(\pm)}=\left\{\alpha_n<\pm x<\alpha_{n+1}\right\}\times \mathbb{R}\,,
$$
and
$$
G_0=(-\alpha_1, \alpha_1)^2\,,
$$
where $\{\alpha_n\}_{n=1}^\infty$ is a monotonically increasing sequence such that $\alpha_n\to\infty$. Let $h_n^{(\pm)}\,, \widetilde{h}_n^{(\pm)}$, and $h_0$ be the Neumann restrictions of the operator $H$ to the regions $G_n^{(\pm)}\,,Q_n^{(\pm)}$ and $G_0$, respectively. Then we have the inequality
\begin{equation}\label{N1}
H\ge\bigoplus_{n=1}^\infty\:\left(h_n^{(+)}\oplus h_n^{(-)} \right) \,\oplus\,
\bigoplus_{n=1}^\infty\left(\widetilde{h}_n^{(+)}\oplus\widetilde{h}_n^{(-)} \right)\,\oplus h_0\,,
\end{equation}
and consequently, $\mathrm{tr}(\omega-H)_+^\sigma$ can be estimated from above by the sum of the corresponding traces of three components of the right-hand side of (\ref{N1}), or their lower bounds. We begin with $h_n^{(+)}$, where the corresponding quadratic form $Q_{h_n^{(+)}}$ estimates as
\begin{eqnarray}\nonumber
\lefteqn{Q_{h_n^{(+)}}[u]=\int_{-\alpha_1}^{\alpha_1}
\int_{\alpha_n}^{\alpha_{n+1}}\left|\frac{\partial{u}}{\partial{x}}\right|^2\,\mathrm{d}x\,\mathrm{d}y+\int_{-\alpha_1}^{\alpha_1}
\int_{\alpha_n}^{\alpha_{n+1}}\left|\frac{\partial{u}}{\partial{y}}\right|^2\,\mathrm{d}x\,\mathrm{d}y} \\ && \nonumber \qquad +\omega^2\int_{-\alpha_1}^{\alpha_1}
\int_{\alpha_n}^{\alpha_{n+1}}y^2
|u|^2\,\mathrm{d}x\,\mathrm{d}y -\lambda
\int_{-\alpha_1}^{\alpha_1} \int_{\alpha_n}^{\alpha_{n+1}}y^2 V(x y)
|u|^2\,\mathrm{d}x\,\mathrm{d}y\\ && \label{Q1}\ge\int_{\alpha_n}^{\alpha_{n+1}}
\biggl(\int_{-\alpha_1}^{\alpha_1}\left|\frac{\partial{u}}{\partial{x}}\right|^2\,\mathrm{d}x+\omega^2\int_{-\alpha_1}^{\alpha_1}y^2
|u|^2\,\mathrm{d}x - \lambda
\int_{-\alpha_1}^{\alpha_1}y^2 V(x y)
|u|^2\,\mathrm{d}x\biggr)\,\mathrm{d}y \phantom{AA}
\end{eqnarray}
for any $u$ from its domain by neglecting the second term in the first expression. For any fixed $y\neq0$ we change of the variables in the inner
integral on the right-hand side of (\ref{Q1}) and denote by $w(t, y)=u\left(\frac{t}{y}, y\right)$. By choosing $\alpha_1=\sqrt{\kappa}$, where $\kappa$ is given by (\ref{kappa}) we arrive at the relation
\begin{eqnarray*}
\lefteqn{\int_{-\alpha_1}^{\alpha_1}\left|\frac{\partial{u}}{\partial{x}}\right|^2(x, y)\,\mathrm{d}y+\omega^2\, y^2\int_{-\alpha_1}^{\alpha_1} |u|^2(x, y)\,\mathrm{d}x-\lambda y^2\int_{-\alpha_1}^{\alpha_1} V(x y) |u|(x, y)^2\,\mathrm{d}x} \\ && =\frac{1}{y}\Big(y^2\int_{-y \alpha_1}^{y \alpha_1}\left|\frac{\partial{w}}{\partial{t}}\right|^2(t, y)\,\mathrm{d}t+\omega^2\,y^2 \int_{-y \alpha_1}^{y \alpha_1} |w|^2(t, y)\,\mathrm{d}t \\ && \qquad -\lambda y^2\int_{-y \alpha_1}^{y \alpha_1}V(t) |w|^2(t, y)\,\mathrm{d}t\Big)\ge \frac{y \gamma_0}{2}\,. \phantom{AAAAAAAAAAAAAAA}
\end{eqnarray*}
This inequality together with (\ref{Q1}) imply that if
$$
\alpha_1=\max\left\{\frac{2\omega}{\gamma_0}\,,\sqrt{\kappa}\right\}\,,
$$
the operators $h_n^{(+)}\,,n=1, 2, \ldots$, have an empty  spectrum below $\omega$, and the same is \emph{mutatis mutandis} true for ${h}_n^{(-)},\,n=1,2,\ldots$.

Let us next pass to the operators $\widetilde{h}_n^{(\pm)},\,n=1, 2, \ldots$. Since the potential $V$ is compactly supported by assumption we have the estimate
\begin{equation}\label{h_n}
\widetilde{h}_n^{(+)}\ge-\Delta+\omega^2 y^2-\frac{\lambda \|V\|_\infty\, a^2}{\alpha_n^2}\,.
\end{equation}
Since the right-hand side of (\ref{h_n}) allows for separation of variables, the spectrum of $h_n^{(+)}$ is the ``sum'' of the spectra of the one-dimensional Neumann operator $-\frac{\mathrm{d}^2}{\mathrm{d}x^2}$ on the interval $(\alpha_n, \alpha_{n+1})$ and the operator $-\frac{\mathrm{d}^2}{\mathrm{d}y^2}+\omega^2 y^2-\frac{\lambda \|V\|_\infty a^2}{\alpha_n^2}$ on $L^2(\mathbb{R})$. Consider first the latter.
Under the assumption
\begin{equation}\label{V}
\frac{\lambda \|V\|_\infty\, a^2}{\alpha_1^2}\le2\omega
\end{equation}
this operator has no more than one eigenvalue below $\omega$  and
\begin{equation}\label{est.alpha}
\left|\omega-\lambda_1^n\right|\le\frac{\lambda \|V\|_\infty\, a^2}{\alpha_n^2}
\end{equation}
holds, where  $\lambda_1^n$ is the indicated eigenvalue.

The spectrum of the one-dimensional Neumann Laplacian on interval $(\alpha_n, \alpha_{n+1})$ consists of simple eigenvalues, $\left\{\frac{\pi^2 j^2}{(\alpha_{n+1}-\alpha_n)^2}\right\}_{j=0}^\infty$, and by choosing
\begin{equation}\label{alpha_n+1}
\alpha_{n+1}-\alpha_n\le\frac{\pi \alpha_1}{\sqrt{\lambda \|V\|_\infty}\, a}
\end{equation}
one can achieve that all the eigenvalues except the one with $j=0$ are not less than $\omega$. Hence we obtain in view of (\ref{est.alpha}) the following estimates,
$$
\mathrm{tr}\left(\omega-\widetilde{h}_n^{(+)}\right)_+^\sigma\le\frac{\lambda^\sigma \|V\|_\infty^\sigma\, a^{2\sigma}}{\alpha_n^{2\sigma}}\,, \quad n=1, 2,\ldots\,,
$$
and
$$
\mathrm{tr}\left(\omega-\bigoplus_{n=1}^\infty \widetilde{h}_n^{(+)}\right)_+^\sigma\le\lambda^{\sigma}\|V\|_\infty^\sigma a^{2\sigma}\sum_{n=1}^\infty \frac{1}{\alpha_n^{2\sigma}}
$$
for any $\sigma\ge0$. Next we are going to minimize the right-hand side of the last inequality. In view of assumptions (\ref{V}) and (\ref{alpha_n+1}) we may choose $\alpha_{n+1}=\alpha_n+\frac{\pi \alpha_1}{\sqrt{\lambda \|V\|_\infty}\, a}=\alpha_1+\frac{n \pi \alpha_1}{\sqrt{\lambda \|V\|_\infty}\, a}$ starting from some $\alpha_1\ge\frac{\sqrt{\lambda \|V\|_\infty}\, a}{\sqrt{2 \omega}}$ which means that
\begin{equation}\label{h_n.estimate}
\mathrm{tr}\left(\omega-\bigoplus_{n=1}^\infty \widetilde{h}_n^{(+)}\right)_+^\sigma\le   \lambda^{2\sigma} \|V\|_\infty^{2\sigma} a^{4\sigma}\sum_{n=1}^\infty\frac{1}{\alpha_1^{2\sigma}\left(\sqrt{\lambda \|V\|_\infty} a+(n-1) \pi\right)^{2\sigma}}\,.
\end{equation}
In the same way one can establish the estimate for operators $\widetilde{h}_n^{(-)}\,,n=1, 2, \ldots$,
\begin{equation}\label{tildeh_n.estimate}
\mathrm{tr}\left(\omega-\bigoplus_{n=1}^\infty \widetilde{h}_n^{(-)}\right)_+^\sigma\le\lambda^{2\sigma} \|V\|_\infty^{2\sigma} a^{4\sigma}\sum_{n=1}^\infty\frac{1}{\alpha_1^{2\sigma}\left(\sqrt{\lambda \|V\|_\infty} a+(n-1) \pi\right)^{2\sigma}}\,.
\end{equation}
Finally, the operator $h_0$ can be estimated from below by the Neumann operator
$$
\widetilde{h}_0=-\Delta-\lambda \alpha_1^2 \|V\|_\infty\quad\text{on}\quad (-\alpha_1, \alpha_1)^2
$$
the spectrum of which is simple and given by $\left\{\frac{\pi^2 (j^2+q^2)}{4\alpha_1^2}-\lambda \alpha_1^2 \|V\|_\infty\right\}_{j, q=0}^\infty$. Consequently,
\begin{equation}\label{LT.ineq}
\mathrm{tr}(\omega-h_0)_+^\sigma\le\mathrm{tr}\left(\omega-\widetilde{h}_0\right)_+^\sigma\le\left(\omega+\lambda \alpha_1^2 \|V\|_\infty\right)^\sigma \left(\frac{2\alpha_1 \sqrt{\omega+\lambda \alpha_1^2 \|V\|_\infty}}{\pi}+1\right)^2,\quad  \sigma>\frac12\,.
\end{equation}
Choosing now  $\alpha_1$ according to (\ref{alpha.}) and using the estimates (\ref{h_n.estimate})-(\ref{LT.ineq}) in combination with the fact that the operators $h_n^{(\pm)}$ have empty spectrum below $\omega$, we conclude the proof of the theorem.\end{proof}

\end{document}